\documentclass[runningheads]{llncs}
\usepackage{todonotes}

\usepackage[utf8]{inputenc}          
\usepackage[T1]{fontenc}             
\usepackage[USenglish]{babel}        
\usepackage[intlimits]{amsmath}      
\usepackage{amsfonts}                
\usepackage{amssymb}                 
\usepackage{amsmath}                 

\usepackage{mathtools}
\usepackage{setspace}
\usepackage{graphicx}
\usepackage{hyperref}

\usepackage{algorithm}
\usepackage[noend]{algpseudocode}

\newcommand{\DS}{\textsc{Dominating Set}}
\newcommand{\EDS}{\textsc{Extension Dominating Set}}
\newcommand{\ERD}{\textsc{Extension Roman Domination}}
\newcommand{\RD}{\textsc{Roman Domination}}
\newcommand{\HS}{\textsc{Hitting Set}}
\newcommand{\NP}{\textsf{NP}}

\newcommand{\RomanUpperbound}{1.9332}
\newcommand{\RomanLowerbound}{1.7441}
\newcommand{\Oh}{\mathcal{O}}

\newenvironment{pf}{\begin{proof}}{\hfill\qed\end{proof}}

\begin{document}
\title{Minimal Roman Dominating Functions: Extensions and Enumeration}
\author{Faisal N. Abu-Khzam\inst1\orcidID{0000-0001-5221-8421} 
\and Henning Fernau\inst2\orcidID{0000-0002-4444-3220}
\and Kevin Mann\inst2\orcidID{0000-0002-0880-2513} 
}
\authorrunning{F. Abu-Khzam et al.}
\institute{
Department of Computer Science and Mathematics\\
Lebanese American University, 
Beirut, Lebanon.\\
\email{faisal.abukhzam@lau.edu.lb}\\
\and
Universit\"at Trier, Fachber.~4 -- Abteilung Informatikwissenschaften\\  
54286 Trier, Germany.\\
\email{\{fernau,mann\}@uni-trier.de}
}
\maketitle 

\begin{abstract}
Roman domination is one of the many variants of domination that keeps most of the 
complexity features of the classical domination problem. We prove that Roman domination behaves differently in two aspects: enumeration and extension. We develop non-trivial enumeration algorithms for minimal Roman domination functions with polynomial delay and polynomial space. Recall that the existence of a similar enumeration result for minimal dominating sets is open for decades.
Our result is based on a polynomial-time algorithm for 
\ERD: Given a graph $G=(V,E)$ and a function $f:V\to\{0,1,2\}$, is there a  minimal Roman domination function $\Tilde{f}$ with $f\leq \Tilde{f}$? Here, $\leq$ lifts 
$0< 1< 2$  pointwise; minimality is understood in this order. Our enumeration algorithm is also analyzed from an input-sensitive viewpoint, leading to a run-time estimate of $\Oh(\RomanUpperbound^n)$ for graphs of order~$n$; this is complemented by a lower bound example of $\Omega(\RomanLowerbound^n)$.

\keywords{Roman domination  \and Extension problems \and Enumeration.}
\end{abstract}

\section{Introduction}

This paper combines four lines of research: (a) studying variations of domination problems, here the Roman domination~\cite{Cocetal2004,Dre2000a,HHS98}; (b) input-sensitive enumeration of minimal solutions, a topic that has drawn attention in particular from people also interested in domination problems \cite{AbuHeg2016,CouHHK2013,CouLetLie2015,GolHKKV2016,GolHegKra2016}; (c) related to (and motivated by) enumeration, extension problems have been introduced and studied in particular in the context of domination problems\footnote{Historically, a logical extension problem~\cite{BorGurHam98} should be mentioned, as it has led to \cite[Th\'eor\`eme 2.16]{Mar2013a}, dealing with an extension variant of 3-\HS; also see \cite[Proposition 3.39]{Mar2013a} concerning implications for \EDS.} in \cite{Bazetal2018,BonDHR2019,CasFKMS2019a,CasFGMS2021,KanLMNU2015,KanLMNU1516,Mar2013a}: is a given set a subset of any minimal dominating set?; 
(d) the \textsc{Hitting Set Transversal Problem} is the question if all minimal hitting sets of a hypergraph can be enumerated with polynomial delay (or even output-polynomial) only: this question is open for four decades by now and is equivalent to several enumeration problems in logic, database theory and also to enumerating minimal dominating sets in graphs, see \cite{CreKPSV2019,EitGot95,GaiVer2017,KanLMN2014}. By way of contrast, we show that enumerating all minimal Roman domination functions is possible with polynomial delay, a result which is  quite surprising in view of the general similarities between the complexities of domination and Roman domination problems.

\begin{figure}[bth]
\begin{center}
\includegraphics[width=.65\textwidth]{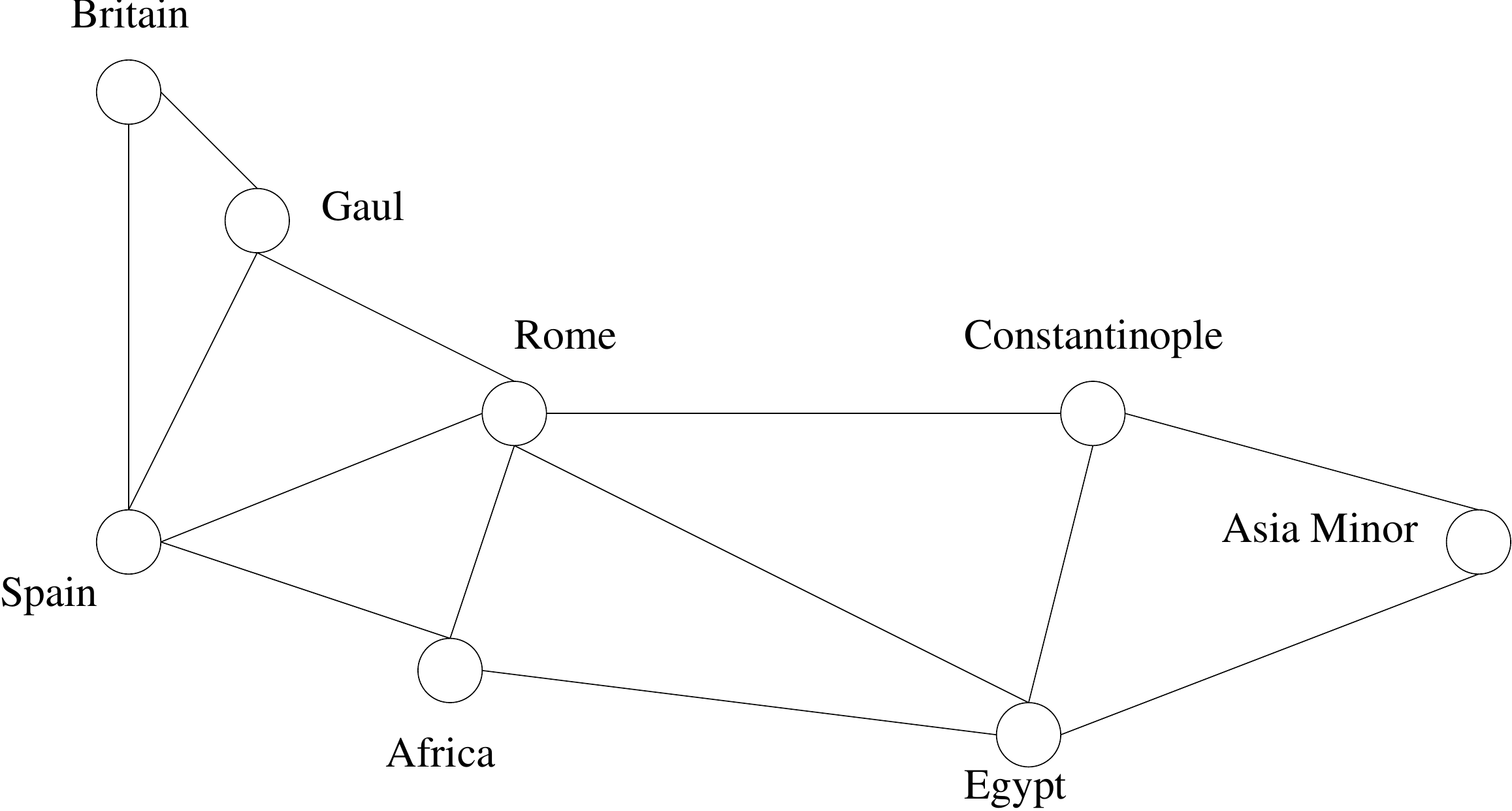}
\end{center}
\caption{\label{fig-Roman-map}The Roman Empire in the times of Constantine}
\end{figure}

\RD\ comes with a nice (hi)story: 
 namely, it should reflect the
idea of how to secure the Roman Empire by positioning the armies (legions)
on the various parts of the Empire in a way that either
(1)  a specific region $r$ is also the location of at least one army or 
(2) one region $r'$  neighboring $r$
has two armies, so that $r'$ can afford sending off one army to the region $r$ 
 (in case of an attack) without diminishing self-defense capabilities.
 More specifically, Emperor 
Constantine had a look at a map of his empire (as discussed in~\cite{Ste99}, also see Fig.~\ref{fig-Roman-map}).\footnote{
The historical background is also nicely described
in the online Johns Hopkins Magazine,  visit 
\url{http://www.jhu.edu/~jhumag/0497web/locate3.html} to pre-view~\cite{ReVRos2000}.} 
Related is the island hopping strategy pursued by General MacArthur in World War II in the Pacific theater to gradually increase the 
US-secured areas. 

\RD\ has received a lot of attention from the algorithmic community in the past 15 years~\cite{Ben2004,ChaCCKLP2013,Dre2000a,Fer08,Lie2007,Lieetal2008,LiuCha2013,Pagetal2002,PenTsa2007,ShaWanHu2010}.
Relevant to our paper is the development of exact algorithms for \RD: combining ideas from \cite{Lie2007,Roo2011}, an $\mathcal{O}(1.5014^n)$ exponential-time and -space algorithm (making use of known \textsc{Set Cover} algorithms via a transformation to \textsc{Partial Dominating Set})  was presented in~\cite{ShiKoh2014}. 
 In \cite{Chaetal2009,CheHHHM2016,Favetal2009,HedRSW2013,KraPavTap2012,LiuCha2012,LiuCha2012a,MobShe2008,XinCheChe2006,XueYuaBao2009,YerRod2013a}, more combinatorial studies can be found. This culminated in a chapter on Roman domination, stretching over nearly 50 pages in the monograph~\cite{HayHedHen2020}. There is also an interesting link to the notion of a \emph{differential} of a graph, introduced in~\cite{Masetal2006}, see \cite{BerFerSig2014}, also adding further algorithmic thoughts, as expressed in \cite{AbuBCF2016,BerFer2014,BerFer2015}. For instance, in~\cite{BerFer2014} an  exponential-time algorithm was published, based on a direct Measure-and-Conquer approach.
 
One of the ideas leading to the development of the area of \emph{extension problems} (as described in~\cite{CasFGMS2021}) was  to cut branches of search trees as early as possible, in the following sense:
to each node of the search tree, a so-called pre-solution~$U$ can be associated,  and it is asked if it is possible
to extend $U$ to a meaningful solution~$S$. In the case of \DS, this means that $U$ is   a set of vertices and a `meaningful solution' is an inclusion-wise minimal dominating set. Notice that such a strategy would work not only for computing smallest dominating sets, but also for computing largest minimal dominating set, or for counting minimal solutions, or for enumerating them. 
Alas, as it has been shown by many examples, extension problems turn out to be quite hard problems.
Even for combinatorial problems whose standard decision version is solvable in polynomial time (for instance, \textsc{Edge Cover}), its extension variation is \NP-hard. In such a case, the approach might still be viable, as possibly parameterized algorithms exist with respect to the parameter `pre-solution size'.
This would be interesting, as this parameter is small when a big gain can be expected in terms of an early abort of a search tree branch.
In particular for \EDS, this hope is not fulfilled. To the contrary, with this parameterization $|U|$, \EDS\ is one of the few problems known to be complete for the parameterized complexity class \textsf{W}[3], as shown in~\cite{BlaFLMS2019}.

With an appropriate definition of the notion of minimality, \RD\ becomes one of the few examples where the hope seeing extension variants being efficiently solvable turns out to be true, as we will show in this paper. This is quite a surprising result, as in nearly any other way, \RD\ behaves most similar to \DS.
Together with its combinatorial foundations (a characterization of minimal Roman domination functions), this constitutes the first main result of this paper. 
The main algorithmic exploit of this result is a non-trivial polynomial-space enumeration algorithm for minimal Roman domination functions that guarantees polynomial delay only, which is the second main result of the paper. As mentioned above, the corresponding question
for enumerating minimal dominating sets is open since decades, and we are not aware of any other modification of the concept of domination that seems to preserve any other of the difficulties of \DS, like classical or parameterized or approximation complexities, apart from the complexity of extension and enumeration.
Our enumeration algorithm is a branching algorithm that we analyzed with a simple Measure \& Conquer approach, yielding a running time  of $\Oh(\RomanUpperbound^n)$, which also gives an upper bound on the number of minimal Roman dominating functions of an $n$-vertex graph. This result is complemented by a simple example that proves a lower bound of $\Omega(\RomanLowerbound^n)$ for the number of minimal Roman dominating functions on graphs of order~$n$.

\section{Definitions}

Let $\mathbb{N}=\{1,2,3,\dots\}$ be the set of positive integers. For $n\in\mathbb{N}$, let $[n]=\{m\in\mathbb{N}\mid m\leq n\}$.
We only consider undirected simple graphs. 
Let $G=\left(V,E\right)$ be a graph. For $U\subseteq V$, $G[U]$ denotes the graph induced by~$U$. 
For $v\in V$, $N_G(v)\coloneqq\{u\in V\mid \lbrace u,v\rbrace\in E\}$ denotes the \emph{open neighborhood} of~$v$, while $N_G[v]\coloneqq N_G(v)\cup\{v\}$ is the \emph{closed neighborhood}  of~$v$. We extend such set-valued functions $X:V\to 2^V$ to $X:2^V\to 2^V$ by setting $X(U)=\bigcup_{u\in U}X(u)$. Subset $D\subseteq V$ is a  \emph{dominating set}, or ds for short, if $N_G[D]=V$. 
For $D\subseteq V$ and $v\in D$, define the \emph{private neighborhood} of $v\in V$ with respect to~$D$ as $P_{G,D}\left( v\right)\coloneqq N_G\left[ v\right] \setminus N_G\left[D\setminus \lbrace v\rbrace\right]$.
A function $f\colon V \to \lbrace 0,1,2 \rbrace$ is called a \emph{Roman dominating function}, or rdf for short, if for each $v\in V$ with $f\left(v\right) = 0$, there exists a $u\in N_G\left( v \right)$ with $f\left(u\right)=2$. 
To simplify the notation, we define $V_i\left(f\right)\coloneqq \lbrace v\in V\mid f\left( v\right)=i\rbrace$ for $i\in\lbrace0,1,2\rbrace$. The \emph{weight} $w_f$ of a function $f\colon V \to \lbrace 0,1,2 \rbrace$ equals $|V_1\left(f\right)|+2|V_2\left(f\right)|$. The classical \RD\ problem asks, given $G$ and an integer $k$, if there exists an rdf for~$G$  of weight at most~$k$. Connecting to the original motivation, $G$ models a map of regions, and if the region vertex~$v$ belongs to~$V_i$, then we place $i$ armies on~$v$.

For the definition of the problem \textsc{Extension Roman Domination}, we need to define the order $\leq$ on $\lbrace 0,1,2\rbrace^{V}$ first:
for $f,g \in \lbrace 0,1,2\rbrace^{V}$, let $f\leq g$ if and only if $f\left(v\right)\leq g\left(v\right)$ for all $v\in V$. In other words, we extend the usual linear ordering $\leq$ on $\{0,1,2\}$ to functions mapping to $\{0,1,2\}$ in a pointwise manner. 
 We call a function $f\in \lbrace 0,1,2\rbrace^{V}$ a \emph{minimal Roman dominating function} if and only if $f$ is a rdf and there exists no rdf $g$, $g\neq f$, with $g\leq f$.\footnote{According to \cite{HayHedHen2020}, this notion of minimality for rdf was coined by Cockayne but then dismissed, as it does not give a proper notion of \emph{upper Roman domination} number. However, in our context, this definition seems to be the most natural one, as it also perfectly fits the extension framework proposed in \cite{CasFGMS2022}. We will propose in \autoref{sec:alternative-notion} yet another notion of minimal rdf that also  fits the mentioned extension framework.} The weights of minimal rdf can vary considerably. Consider for example a star $K_{1,n}$ with center~$c$. Then, $f_1(c)=2$, $f_1(v)=0$ otherwise; $f_2(v)=1$ for all vertices~$v$; $f_3(c)=0$, $f_3(u)=2$ for one $u\neq c$, $f_3(v)=1$ otherwise, define three minimal rdf with weights $w_{f_1}=2$,  and $w_{f_2}=w_{f_3}=n+1$. 

\vspace{5pt}

\noindent
\centerline{\fbox{\begin{minipage}{.99\textwidth}
\textbf{Problem name: }\ERD, or \textsc{ExtRD} for short\\
\textbf{Given: } A graph $G=\left( V,E\right)$ and a function $f\in \lbrace 0,1,2 \rbrace^V.$\\
\textbf{Question: } Is there a minimal rdf $\widetilde{f} \in \lbrace 0,1,2\rbrace^V$ with $f\leq \widetilde{f}$?
\end{minipage}
}}

\vspace{5pt}

As our first main result, we are going to show that \textsc{ExtRD} can be solved in polynomial time in \autoref{sec:poly-time-ExtRD}.
To this end, we need some understanding of the combinatorial nature of this problem, which we provide in \autoref{sec:properties-minimal-rdf}.

The second problem that we consider is that of enumeration, both from an output-sensitive and from an input-sensitive perspective.

\vspace{5pt}

\noindent
\centerline{\fbox{\begin{minipage}{.99\textwidth}
\textbf{Problem name: }\textsc{Roman Domination Enumeration}, or  \textsc{RDEnum} for short\\
\textbf{Given: } A graph $G=\left( V,E\right)$.\\
\textbf{Task: } Enumerate all minimal rdf ${f} \in \lbrace 0,1,2\rbrace^V$ of~$G$!
\end{minipage}
}}

From an output-sensitive perspective, it is interesting to perform this enumeration without repetitions and with polynomial delay, which means that there is a polynomial $p$ such that between the consecutive outputs of any two minimal rdf of a graph of order~$n$ that are enumerated, no more than $p(n)$ time elapses, including the corner-cases at the beginning and at the end of the algorithm. From an input-sensitive perspective, we want to upper-bound the running time of the algorithm, measured against the order of the input graph. The obtained run-time bound should not be too different from known lower bounds, given by graph families where one can prove that a certain number of minimal rdf must exist. 
Our algorithm will be analyzed from both perspectives and achieves both goals.
This is explained in \autoref{sec:enum-minimal-rdf-simple} and in \autoref{sec:enum-minimal-rdf-refined}.

\section{Properties of Minimal Roman Dominating Functions}
\label{sec:properties-minimal-rdf}

\begin{theorem}\label{t_1_2_neigborhood}
Let $G=\left(V,E\right)$ be a graph and $f: \: V \to \lbrace 0,1,2\rbrace$ be a minimal rdf. Then $N_G\left[V_2\left(f\right)\right]\cap V_1\left(f\right)=\emptyset$ holds.
\end{theorem}
\begin{pf}
Assume that there exists a $\lbrace u,v\rbrace\in E$ with $f\left(v\right) = 2$ and $f\left(u\right)=1$.
Let
$$ \widetilde{f}:V\to \lbrace 0,1,2\rbrace,\: w\mapsto\begin{cases}
f\left(w\right), &w\neq u\\
0,& w=u
\end{cases}$$
We show that $\widetilde{f}$ is a rdf, which contradicts  the minimality of $f$, as $\widetilde{f}\leq f$ and $\widetilde{f}\left( u \right) < f\left( u\right)$ are given by construction.
Consider $w \in V_0\left(\widetilde{f}\right)$. If $w= u$, $w$ is dominated by $v$, as $\lbrace u,v\rbrace\in E$. Consider $w\neq u$. Since $f$ is a rdf and $V_0\left(f\right)\cup\lbrace u\rbrace = V_0\left(\widetilde{f}\right)$, there exists a $t\in N_G\left[ w\right]\cap V_2\left(f\right)$. By construction of $\widetilde{f}$, $V_2\left(f\right)=V_2\left(\widetilde{f}\right)$ holds. This implies $N_G\left[ w\right]\cap V_2\left(\widetilde{f}\right)\neq \emptyset$.  Hence, $\widetilde{f}$ is a rdf. 
\end{pf}

\begin{theorem}\label{t_private_neighborhood}
Let $G=\left(V,E\right)$ be a graph and $f: \: V \to \lbrace 0,1,2\rbrace$ be a minimal rdf. Then for all $v\in V_2\left( f \right)$,  $P_{G\left[V_0\left(f\right) \cup  V_2\left(f\right)\right], V_2\left(f\right)}\left(v\right) \nsubseteq \lbrace v \rbrace$ holds.
\end{theorem}
\begin{pf}
Define $G'\coloneqq G\left[V\setminus V_1\left( f \right)\right] = G\left[V_0\left(f\right) \cup  V_2\left(f\right)\right]$. In contrast to the claim, assume that there exists a $v\in V_2\left( f \right)$ with $P_{G', V_2\left(f\right)}(v) \subseteq \lbrace v \rbrace$. Define 
$$ \widetilde{f}:V\to \lbrace 0,1,2\rbrace,\: w\mapsto\begin{cases}
f\left(w\right), &w\neq v\\
1,& w=v
\end{cases} $$
We show that $\widetilde{f}$ is a rdf, which contradicts  the minimality of $f$, as $\widetilde{f}\leq f$ and $\widetilde{f}\left( v \right) < f\left( v\right)$ are given by construction.
Let $u\in V_0\left(\widetilde{f}\right)=V_0\left({f}\right)$. We must show that some neighbor of $u$ belongs to $V_2\left(\widetilde{f}\right)=V_2\left(f\right) \setminus \lbrace v\rbrace$.
Then, $\widetilde f$ is a rdf.

First, assume that $u$ is a neighbor of~$v$. By the choice of~$v$, $u$ is not a private neighbor of~$v$. Hence, there exists a $w\in N_{G}\left[u\right] \cap\left( V_2\left(f\right) \setminus \lbrace v\rbrace\right) = N_{G}\left[u\right]\cap V_2\left(\widetilde{f}\right)$. Secondly, 
if $u\in V_0\left(\widetilde{f}\right)$ is not a neighbor of $v$, then there exists a $w\in V_2\left(f\right)\setminus\{v\}$ that dominates~$u$, i.e.,  $w\in N_{G}\left[u\right]\cap \left(V_2\left(f\right)\setminus \lbrace v\rbrace\right) =  N_{G}\left[u\right]\cap V_2\left(\widetilde{f}\right)$.
\end{pf}

\noindent
As each $v\in V_0\left(f\right)$ has to be dominated by a $w\in V_2\left(f\right)$, the next claim follows.

\begin{corollary}\label{c_min_dom}
Let $G=\left(V,E\right)$ be a graph and $f\in \lbrace 0,1,2\rbrace^V$ be a minimal rdf. Then, $V_2\coloneqq V_2\left(f\right)$ is a minimal ds of $G\left[ N_G[V_2]\right]$, with 
$N_G[V_2]=V_0\left(f\right)\cup V_2$.
\end{corollary}

\begin{remark}
We can generalize the last statement as follows: Let $G=\left(V,E\right)$ be a graph and $f: \: V \to \lbrace 0,1,2\rbrace$ be a minimal rdf. Let $I\subseteq V_1(f)$ be an independent set in $G$. Then, $V_2\left(f\right)\cup I$ is a minimal ds of $G\left[ V_0\left(f\right)\cup V_2\left(f\right)\cup I\right]$. If $I$ is a maximal independent set in $G[V_1(f)]$, then $V_2\left(f\right)\cup I$ is a minimal ds of $G\left[ V_0\left(f\right)\cup V_2\left(f\right)\cup V_1(f)\right]$.
\end{remark}

\noindent
This allows us to deduce the following characterization result.

\begin{theorem}\label{t_porperty_min_rdf}
Let $G=\left(V,E\right)$ be a graph, $f: \: V \to \lbrace 0,1,2\rbrace$ and abbreviate
$G'\coloneqq G\left[ V_0\left(f\right)\cup V_2\left(f\right)\right]$. Then, $f$ is a minimal rdf if and only if the following conditions hold:
\begin{enumerate}
\item$N_G\left[V_2\left(f\right)\right]\cap V_1\left(f\right)=\emptyset$,\label{con_1_2}
\item $\forall v\in V_2\left(f\right) :\: P_{G',V_2\left(f\right)}\left( v \right) \nsubseteq \lbrace v\rbrace$, also called \emph{privacy condition}, and \label{con_private}
\item $V_2\left(f\right)$ is a minimal dominating set of $G'$.\label{con_min_dom}
\end{enumerate}
\end{theorem}
\begin{pf}
The ``only if'' follows by  \autoref{t_1_2_neigborhood}, \autoref{t_private_neighborhood} and  \autoref{c_min_dom}. 

Let $f$ be a function that fulfills the three conditions. Since $V_2\left(f\right)$ is a dominating set on $G'$, for each $u\in V_0\left( f\right)$, there exists a $v\in V_2\left(f\right)\cap N_{G}\left[u\right]$. Therefore, $f$ is a rdf. 
Let $\widetilde{f}:V \to \lbrace 0,1,2 \rbrace$ be a minimal rdf with $\widetilde{f}\leq f$. Therefore, $\widetilde{f}$ (also) satisfies the three conditions by  \autoref{t_1_2_neigborhood}, \autoref{t_private_neighborhood} and \autoref{c_min_dom}. 
Assume that there exists a $v\in V$ with $\widetilde{f}\left( v\right) < f\left( v \right)$. Hence,  $V_2\left(\widetilde{f}\right)\subseteq V_2\left(f\right)\setminus \lbrace v\rbrace$.  \\
\textbf{Case 1:} $\widetilde{f}\left( v\right)=0, f\left( v \right) =1$. Therefore, there exists a $u\in N_G\left(v\right)$ with $f\left(u\right)\geq \widetilde{f}\left(u\right)=2$. This contradicts Condition~\ref{con_1_2}.\\
\textbf{Case 2:} $\widetilde{f}\left( v\right)\in\lbrace 0, 1\rbrace, f\left( v \right) =2$.  Let  $u\in N_G\left(v\right)$ with $f(u)=0$. This implies $\widetilde{f}(u)=0$ and
$$\emptyset \neq  N_G\left[u\right] \cap V_2\left(\widetilde{f}\right)\subseteq N_G\left[u\right] \cap V_2\left(f\right)\setminus \lbrace v \rbrace$$ holds. Therefore, $N_G\left( v \right) \subseteq N_G\left[ V_2\left(f\right) \setminus \lbrace v\rbrace\right]$. This contradicts Condition~\ref{con_private}.

Thus, $\widetilde{f}=f$ holds and $f$ is minimal. 
\end{pf}

\noindent
We conclude this section with an upper bound on the size of $V_2(f)$.

\begin{lemma}\label{lem:V2-bound}
Let $G=\left(V,E\right)$ be a graph and $f: V \to \lbrace0,1,2\rbrace$ be a minimal rdf. Then $2 \: \vert V_2\left(f\right)\vert \leq \vert V \vert  $ holds.
\end{lemma}

\begin{pf}
Consider a graph $G=\left(V,E\right)$ and a minimal rdf $f:V\to\lbrace 0,1,2\rbrace$. For each $v\in V_2\left(f\right)$, let $P_f\left(v\right)= P_{G\left[V_0\left(f\right) \cup  V_2\left(f\right)\right], V_2\left(f\right)}\left(v\right) \setminus \lbrace v \rbrace\subseteq V\setminus V_2\left(f\right) $. By \autoref{t_private_neighborhood}, these sets are not empty and, by definition, they do not intersect.  Hence, we get:
$$ \vert V \vert =  \vert V_2\left( f\right) \vert +\vert V\setminus V_2\left( f\right) \vert \geq \vert V_2\left( f\right) \vert +\left\vert\bigcup_{v\in V_2\left(f\right)}P_f\left(v\right) \right\vert\geq 2\: \vert V_2\left( f\right) \vert\,. $$ 
Therefore, the claim is true. 
\end{pf}

\section{A Polynomial-time Algorithm for \textsc{ExtRD}}
\label{sec:poly-time-ExtRD}

With \autoref{t_porperty_min_rdf}, we can construct an algorithm that solves the problem \textsc{Extension Roman domination} in polynomial time. 
\begin{algorithm}
\caption{Solving instances of \textsc{ExtRD}}\label{alg}
\begin{algorithmic}[1]
\Procedure{ExtRD Solver}{$G,f$}\newline
 \textbf{Input:} A graph $G=\left(V,E\right)$ and a function $f\colon V\to \lbrace0,1,2\rbrace$.\newline
 \textbf{Output:} Is there a minimal Roman dominating function $\widetilde{f}$ with $f\leq \widetilde{f}$?
\State $\widetilde{f}\coloneqq f$. \label{alg_init}
\State $M_2\coloneqq V_2\left(f\right)$. \{ Invariant: $M_2=V_2(\widetilde{f})$ \} \label{alg_invariant}
\State $ M \coloneqq M_2$. \{ All $v\in V_2(\widetilde{f})$ are considered below; invariant: $M\subseteq M_2$. \}
\label{alg_before_while}
\While{$M\neq \emptyset$}
\State Choose $v\in M$. \{ Hence, $\widetilde{f}(v)=2$. \}
\For {$u\in N\left(v\right)$}
\If {$\widetilde{f}\left(u\right)=1$}\label{alg_if_no}
\State $\widetilde{f}\left(u\right)\coloneqq 2$.
\State Add $u$ to $M$ and to $M_2$.
\EndIf
\EndFor
\State Delete $v$ from $M$. 
\EndWhile
\For{$v\in M_2$}\label{alg_for_no}
\If{$N_G\left(v\right)\subseteq N_G\left[M_2\setminus \lbrace v\rbrace\right]$}\label{alg_private_test}
\State \textbf{Return No}. \label{alg_no}
\EndIf
\EndFor
\For{$v\in V\setminus N_G\left[ M_2\right]$}\label{alg_fil_for}
\State $\widetilde{f}\left(v\right)\coloneqq 1$.
\EndFor 
\State\textbf{Return Yes}.
\EndProcedure
\end{algorithmic}
\end{algorithm}

\begin{theorem}\label{theorem:correctness_alg}
Let $G=\left(V,E\right)$ be a graph and $f\colon V\to\lbrace0,1,2\rbrace$. 
For the inputs $G,f$, Algorithm~\ref{alg} returns yes if and only if $\left( G,f\right)$ is a yes-instance of \textsc{ExtRD}. In this case, the function~$\widetilde f$ computed by Algorithm~\ref{alg} is a minimal rdf.
\end{theorem}
\begin{pf}First observe that the invariants stated in Lines~\ref{alg_invariant} and~\ref{alg_before_while} of Algorithm~\ref{alg} are true whenever entering or leaving the while-loop.

Let the answer of the algorithm be \emph{yes} and $\widetilde{f}$ be the function computed by the algorithm. We will show that~$\widetilde{f}$
satisfies the 
conditions formulated in \autoref{t_porperty_min_rdf}.

Observing the if-condition in Line~\ref{alg_if_no}, clearly after the while-loop, no neighbor~$u$ of $v\in V_2\left(\widetilde{f}\right)$ fulfills $\widetilde{f}\left( u\right)=1$. Hence, $\widetilde{f}$ satisfies Condition~\ref{con_1_2}.
If the function $\widetilde{f}$ would contradict Condition~\ref{con_private} of \autoref{t_porperty_min_rdf}, then we would get to Line~\ref{alg_no} and the algorithm would answer \emph{no}. As we are considering a \emph{yes}-answer of our algorithm, we can assume that this privacy condition holds after the for-loop of Line~\ref{alg_for_no}. 
We also can assume that $M_2=V_2\left(\widetilde{f}\right)$ is a minimal ds of the graph $G\left[N_G\left[M_2\right]\right]$. Otherwise, for such a $v\in M_2$ and each $u\in N_G\left(v\right)$, there would exist a $w\in N_G\left[u\right]\cap\left( M_2\setminus\lbrace v\rbrace\right)$. In this case, the algorithm would return \emph{no} in Line~\ref{alg_no}.
In the for-loop of Line~\ref{alg_fil_for}, we update for all $v\in V\setminus N_G\left[ M_2\right]$ the value $\widetilde{f} \left( v\right)$ to~$1$. With the while-loop, this implies $N_G\left[M_2\right] = V_0\left(\widetilde{f}\right)\cup V_2\left(\widetilde{f}\right)$. Therefore, $V_2\left(\widetilde{f}\right)$ is a minimal ds of $G\left[V_0\left(\widetilde{f}\right)\cup V_2\left(\widetilde{f}\right)\right]$. 
Since we do not update the values of $\widetilde{f}$ to two in this last for-loop, Condition~\ref{con_private} from  \autoref{t_porperty_min_rdf} holds. By the while-loop and the for-loop starting in Line~\ref{alg_fil_for}, it is trivial to see that Condition~\ref{con_1_2} also holds for the final $\widetilde{f}$.
We can now use \autoref{t_porperty_min_rdf} to see that $\widetilde{f}$ is a minimal rdf. 

Since we never decrease $\widetilde{f}$ in this algorithm, starting with $\widetilde f=f$ in Line~\ref{alg_init}, we get $f\leq \widetilde{f}$. Therefore, $\left(G,f\right)$ is a \emph{yes}-instance of \textsc{ExtRD}.

Now we assume that $\left(G,f\right)$ is a \emph{yes}-instance, but the algorithm returns \emph{no}. Therefore, there exists a minimal rdf $\overline{f}$ with $f\leq \overline{f}$. Since $N_G\left[V_2\left(\overline{f}\right)\right]\cap V_1\left(\overline{f}\right)=\emptyset$, $\widetilde{f}\leq \overline{f}$ holds for the function~$\widetilde{f}$ in Line~\ref{alg_for_no}. This implies $M_2=V_2\left( \widetilde{f}\right)\subseteq V_2\left(\overline{f}\right)$. 
The algorithm returns \emph{no} if and only if there exists a $v\in M_2$ with 
$$ N_G\left(v\right)\subseteq N_G\left[ M_2\setminus \lbrace v \rbrace\right]\subseteq N_G\left[V_2\left(\overline{f}\right)\setminus \lbrace v\rbrace\right].$$
Applying again Theorem~\ref{t_porperty_min_rdf}, we see that $\overline{f}$ cannot be a minimal rdf, contradicting our assumption.
\end{pf}

\noindent
In \autoref{propos:runtime}, we prove that our algorithm needs polynomial time only.

\begin{proposition}\label{propos:runtime}
Algorithm~\ref{alg} runs in time cubic in the order of the input graph.
\end{proposition}
\begin{pf} Let $G=(V,E)$ be the input graph. 
Define $n=\vert V\vert$. Up to Line~\ref{alg_before_while}, the algorithm can run in linear time. As each vertex can only be once in $M$ and we look at the neighbors of each element in $M$, the while-loop runs in time $\mathcal{O}\left(n^2\right)$. In the for-loop starting in Line~\ref{alg_for_no},  we build for all $v\in M_2$ the set $N_G\left[M_2\setminus \lbrace v\rbrace\right]$. This needs $\mathcal{O}\left( n^3\right)$ time. The other steps of this loop run in time $\mathcal{O}\left(n^2\right)$.
The last for-loop requires linear time. Hence, the algorithm runs in time $\mathcal{O}\left(n^3\right)$.   
\end{pf}

\section{Enumerating  Minimal RDF  for General  Graphs}
\label{sec:enum-minimal-rdf-simple}
 
For general graphs, our general combinatorial observations allow us to strengthen the (trivial) $\mathcal{O}^*(3^n)$-algorithm for enumerating all minimal rdf  for graphs of order~$n$
down to $\mathcal{O}^*(2^n)$, as displayed in Algorithm~\ref{alg:enum}.
To understand the correctness of this enumeration algorithm, the following lemma is crucial.

\begin{algorithm}
\caption{A simple enumeration algorithm for minimal rdf}\label{alg:enum}
\begin{algorithmic}[1]
\Procedure{RD Enumeration}{$G$}\newline
 \textbf{Input:} A graph $G=\left(V,E\right)$.\newline
 \textbf{Output:} Enumeration of all  minimal rdf $f:V\to\{0,1,2\}$.
\For {all functions $f:V\to\{1,2\}$}
\For {all $v\in V$ with $f(v)=1$}
\If {$\exists u\in N_G(v): f(u)=2$}
\State $f(v)\coloneqq 0$.
\EndIf
\EndFor
\State Build graph $G'$ induced by $f^{-1}(\{0,2\})=V_0(f)\cup V_2(f)$.
\State $\text{private-test}\coloneqq 1$.
\For {all $v\in V$ with $f(v)=2$}
\If {$P_{G',V_2(F)}(v)\subseteq\{v\}$}
\State $\text{private-test}\coloneqq 0$.
\EndIf
\EndFor
\If{$\text{private-test}=1$ and if $f^{-1}(2)=V_2(f)$ is a minimal ds of $G'$}
\State Output the current function  $f:V\to\{0,1,2\}$.
\EndIf
\EndFor
\EndProcedure
\end{algorithmic}
\end{algorithm}

\begin{lemma}\label{lem:extend2}
Let $G=(V,E)$ be a graph with $V_2\subseteq V$ such that $P_{G,V_2}\left( v \right) \nsubseteq \lbrace v\rbrace$ for each $v\in V_2$ holds. Then there exists exactly one minimal rdf $f\in \lbrace 0,1,2 \rbrace ^ V$ with $V_2=V_2\left(f\right)$. Algorithm~\ref{alg} can calculate $f$.
\end{lemma}
\begin{pf}
Define \\[-4.5ex]
$$f:V\to \lbrace 0, 1, 2\rbrace, v\mapsto \begin{cases}
2, & v \in V_2\\
1, & v \notin N\left[ V_2\right]\\
0, & \text{otherwise}
\end{cases}$$
Hence, $N_G\left[ V_2 \right]=V_2\cup V_0\left(f\right)$. With the assumption $P_{G,V_2}\left( v \right) \nsubseteq \lbrace v\rbrace$, $V_2$ is a minimal ds of $G[V_2\cup V_0\left(f\right)]$. Furthermore,  $N_G\left[V_2\right]\cap V_1\left(f\right)=\emptyset$. As $V_2=V_2\left(f\right)$, all conditions of \autoref{t_porperty_min_rdf} hold and $f$ is a minimal rdf.

Let $\widetilde{f}\in\lbrace 0,1,2\rbrace^V$ be a minimal rdf with $V_2= V_2\left(\widetilde{f}\right)$. If there exists some $v\in V_0\left( f \right) \cap V_1\left(\widetilde{f}\right)$, this contradicts Condition~\ref{con_1_2}, as $v\in N_G\left[V_2\right] = N_G\left[V_2\left(\widetilde{f}\right)\right]$. Therefore, $ V_0\left( f\right)\subseteq V_0\left( \widetilde{f}\right) $ holds. By the assumption that $\widetilde{f}$ is a rdf, for each $v\in V_0\left( \widetilde{f} \right)$ there exists a $u\in V_2\left(\widetilde{f}\right)\cap N\left[v\right] = V_2\cap N\left[v\right]$. This implies $v\in N_G\left[V_2\right]\setminus V_2= V_0\left( f\right)$. Therefore, $ V_0\left( f\right) = V_0\left( \widetilde{f}\right)$ holds. This implies $f=\widetilde{f}$.

\vspace{5pt}
Define: 
$$\widehat{f}:V\to \lbrace 0, 1, 2\rbrace, v\mapsto \begin{cases}
2, & v\in V_2\\
0, & v\notin V_2
\end{cases}.$$
It is trivial to see that $\widehat{f}\leq f$. By \autoref{theorem:correctness_alg}, Algorithm~\ref{alg} returns \emph{yes} for the input $\widehat{f}$. Let $\overline{f}$ be the the minimal rdf produced by Algorithm~\ref{alg}, given~$\widehat{f}$. We want to show that $V_2=V_2\left( \overline{f} \right)$. We do this by looking at the steps of the algorithm.
Since $V_1\left( \widehat{f} \right) = \emptyset$, the algorithm never gets into the If-clause in Line~\ref{alg_if_no}. This is the only way to update a vertex to the value 2. Therefore, $V_2= V_2\left(\overline{f}\right)$. 
\end{pf}

\begin{proposition} Let $G=\left(V,E\right)$ be a graph. For minimal rdf $f,g \in \lbrace0,1,2\rbrace ^ V$ with $V_2\left(f\right)=V_2\left(g\right)$, it holds $f=g$.
\end{proposition}

\begin{pf}
By \autoref{t_private_neighborhood}, $V_2\left(f\right)$ fulfills the conditions of \autoref{lem:extend2}. Therefore, there exists a unique minimal rdf $h\in\lbrace0,1,2\rbrace ^ V$ with $V_2\left( h \right)=V_2\left( f \right) = V_2\left( g \right)$. Thus. $f=g=h$ holds.
\end{pf}

Hence, there is a bijection between the minimal rdf of a graph $G=(V,E)$ and subsets $V_2\subseteq V$ that satisfy the condition of \autoref{lem:extend2}.

\begin{proposition}\label{prop:RomanEnum}
All minimal rdf of a graph of order~$n$ can be enumerated in time
$\mathcal{O}^*(2^n)$.
\end{proposition}

\begin{pf}
 Consider Algorithm~\ref{alg:enum}. The running time claim is obvious. The correctness of the algorithm is clear due to \autoref{t_porperty_min_rdf} and \autoref{lem:extend2}.  
\end{pf}

The presented algorithm clearly needs polynomial space only, but it is less clear if it has polynomial delay.
Below, we will present a branching algorithm that has both of these desirable properties, and moreover, its running time is below $2^n$. How good or bad such an enumeration is, clearly also depends on examples that provide a lower bound on the number of objects that are enumerated. The next lemma explains why the upper bounds for enumerating minimal rdf must be bigger than those for enumerating minimal dominating sets.

\begin{lemma}
A disjoint collection of $c$ cycles on five vertices yields a graph of order $n=5c$ that has $(16)^c$ many minimal rdf. 
\end{lemma}

\begin{pf}
Let $C_5$ be a cycle of length 5 with $V\left(C_5\right)=\lbrace v_1,\ldots,v_5 \rbrace$ and $E\left( C_5 \right) = \lbrace\lbrace v_i,v_{i+1}\rbrace\mid i\in [4] \rbrace \cup \lbrace \lbrace v_1,v_5 \rbrace \rbrace$. For a $f\in\lbrace0,1,2\rbrace^{V\left( C_5\right)}$ there are at least the following sixteen possibilities for $\left( f\left( v_1\right),\ldots,f\left( v_5\right)\right)$:
\begin{itemize}
    \item zero occurrences of 2: $(1,1,1,1,1)$;
    \item one occurrence of 2: $(2,0,1,1,0)$ and four more cyclic shifts;
    \item two adjacent occurrences of 2:  $(2,2,0,1,0)$  and four more cyclic shifts;
    \item two non-adjacent occurrences of 2: $(2,0,2,0,0)$   and four more cyclic shifts.
\end{itemize}

Therefore, there are at least 16 minimal rdf on $C_5$. To prove that these are all the minimal rdf, we use Lemma~\ref{lem:V2-bound}, which implies $\vert V_2\left(f\right)\vert\leq \frac{\vert V\left(C_5\right)\vert}{2} <3$. Hence, the number of minimal rdf on $C_5$ is at most $\binom{5}{0}+\binom{5}{1}+\binom{5}{2}=16$. 
\end{pf}

\begin{corollary}
There are graphs of order $n$ that have at least ${\sqrt[5]{16}\,}^n\in\Omega(\RomanLowerbound^n)$ many minimal rdf.
\end{corollary}

We checked with the help of a computer program that there are no other connected graphs of order at most eight that yield (by taking disjoint unions) a bigger lower bound.

\section{A Refined Enumeration Algorithm}
 \label{sec:enum-minimal-rdf-refined}

In this section, we are going to prove the following result, which can be considered as the second main result of this paper.
\begin{theorem}\label{thm:minimal-rdf-enumeration}
There is a polynomial-space algorithm that enumerates all minimal rdf of a given graph of order $n$ with polynomial delay and in time $\mathcal{O}^*(\RomanUpperbound^n)$.
\end{theorem}

Notice that this is in stark contrast to what is known about the enumeration of minimal dominating sets, or, equivalently, of minimal hitting sets in hypergraphs. Here, it is a long-standing open problem if 
minimal hitting sets in hypergraphs can be enumerated with polynomial delay.

The remainder of this section is dedicated to describing the proof of this theorem.

\subsection{A bird's eye view on the algorithm}

As all along the search tree, from inner nodes we branch into the two cases if a certain vertex is assigned $2$ or not, it is clear that (with some care concerning the final processing in leaf nodes) no minimal rdf is output twice. Hence, there is no need for the branching algorithm to store intermediate results to test (in a final step) if any solution was generated twice. Therefore, our algorithm needs only polynomial space, as detailed in \autoref{prop:poly-space} and \autoref{cor:poly-space}.

Because we have a polynomial-time procedure that can test if a certain given pre-solution can be extended to a minimal rdf, we can build (a slightly modified version of) this test into an enumeration procedure, hence avoiding unnecessary branchings. 
Therefore, whenever we start with our binary branching, we know that at least one of the search tree branches will return at least one new minimal rdf. Hence, we will not move to more than $N$ nodes in the search tree before outputting a new minimal rdf, where $N$ is upper-bounded by twice the order of the input graph. This is the basic explanation for the claimed polynomial delay, as detailed in \autoref{prop:poly-delay}.

Let $G=(V,E)$ be a graph.
Let us call a(ny partial) function 

\vspace{-15pt}
\begin{equation*}
\begin{split}
 f: V \longrightarrow \{0,1,2,\overline{1}, \overline{2}\}
\end{split}
\end{equation*}
 a \emph{generalized Roman domination function}, or grdf for short. 
Extending previously introduced notation, let $\overline{V_1}(f) = \{x\in V\mid  f(x) = \overline{1}\}$, and $\overline{V_2}(f) = \{x\in V\mid  f(x) = \overline{2}\}$. A vertex is said to be \emph{active} if it has not been assigned a value (yet) under~$f$; these vertices are collected in the set $A(f)$. Hence, for any grdf $f$, we have the partition $V=A(f)\cup V_0(f)\cup V_1(f)\cup V_2(f)\cup \overline{V_1}(f)\cup \overline{V_2}(f)$. 

After performing a branching step, followed by an exhaustive application of the reduction rules, any grdf~$f$ considered in our algorithm always satisfies the following \textbf{(grdf) invariants}:
\begin{enumerate}
    \item $\forall x\in \overline{V_1}(f)\cup V_0(f)\,\exists y\in N_G(x):y\in V_2(f)$,
    \item $\forall x\in V_2(f):N_G(x)\subseteq \overline{V_1}(f) \cup V_0(f) \cup V_2(f)$,    
    \item $\forall x\in V_1(f):N_G(x)\subseteq \overline{V_2}(f)\cup V_0(f)\cup V_1(f)$,
    \item if $\overline{V_2}(f)\neq\emptyset$, then $A(f)\cup \overline{V_1}(f)\neq \emptyset$.\footnote{This condition assumes that our graphs have non-empty vertex sets.}
\end{enumerate}

For the extension test, we will therefore consider the function $\hat f:V\to\{0,1,2\}$ that is derived from a grdf~$f$ as follows: 
$$\hat f(v)=\begin{cases}0, & \text{if }v\in  A(f)\cup V_0(f)\cup \overline{V_1}(f)\cup \overline{V_2}(f)\\
1, & \text{if }v\in V_1(f)\\
2, & \text{if }v\in V_2(f)
\end{cases}$$

The enumeration algorithm uses a combination of reduction and branching rules, starting with the nowhere defined function $f_\bot$, so that $A(f_\bot)=V$. The schematics of the algorithm is shown in Algorithm~\ref{alg:refined-enum}. 
To understand the algorithm, call an rdf $g$ as \emph{consistent} with a grdf $f$ if $g(v)=2$ implies $v\in A(f)\cup V_2(f)\cup \overline{V_1}(f)$ and $g(v)=1$ implies $v\in A(f)\cup V_1(f)\cup \overline{V_2}(f)$ and $g(v)=0$ implies $v\in A(f)\cup V_0(f)\cup \overline{V_1}(f)\cup \overline{V_2}(f)$.
Below, we start with presenting some reduction rules, which also serve as (automatically applied) actions at each branching step, whenever applicable. The branching itself always considers a most attractive vertex $v$ and either gets assigned~2 or not. 
The running time analysis will be performed with a measure-and-conquer approach. Our simple measure is defined by  $\mu(G,f)=|A(f)|+\omega_1 |\overline{V_1}(f)| + \omega_2 |\overline{V_2}(f)|\leq |V| $
for some constants $\omega_1$ and $\omega_2$ that have to be specified later.

The measure never increases when applying a  reduction rule.

\begin{algorithm}
\caption{A refined enumeration algorithm for minimal rdf}\label{alg:refined-enum}
\begin{algorithmic}[1]
\Procedure{Refined RD Enumeration}{$G,f$}\newline
 \textbf{Input:} A graph $G=\left(V,E\right)$, a grdf $f:V\to\{0,1,2,\overline{1},\overline{2}\}$.\newline
 \emph{Assumption:} There exists at least one   minimal rdf  consistent with $f$.\newline
 \textbf{Output:} Enumeration of all  minimal rdf  consistent with $f$.
\If {$f$ is everywhere defined and $f(V)\subseteq \{0,1,2\}$}
\State Output $f$ and return.
\EndIf
\State \{\,We know that $A(f)\cup \overline{V_1}(f)\neq\emptyset$.\,\}
\State Pick a vertex $v\in A(f)\cup \overline{V_1}(f)$ of highest priority for branching.
\State $f_2\coloneqq f $; $f_2(v)\coloneqq 2$.
\State Exhaustively apply reduction rules to $f_2$. \{\,Invariants are valid for $f_2$.\,\}
\If{$\textsc{GenExtRD Solver}\left(G,\widehat{f_2},\overline{V_2}(f_2)\right)$}
\State $\textsc{Refined RD Enumeration}\left(G,{f_2}\right)$.
\EndIf
\State $f_{\overline{2}}\coloneqq f $; \textbf{if} $v\in A(f)$ \textbf{then} $f_{\overline{2}}(v)\coloneqq {\overline{2}}$ \textbf{else} $f_{\overline{2}}(v)\coloneqq 0$.
\State Exhaustively apply reduction rules to $f_{\overline{2}}$. \{\,Invariants are valid for $f_{\overline{2}}$.\,\}
\If{$\textsc{GenExtRD Solver}\left(G,\widehat{f_{\overline{2}}},\overline{V_2}(f_{\overline{2}})\right)$}
\State $\textsc{Refined RD Enumeration}\left(G,{f_{\overline{2}}}\right)$.
\EndIf
\EndProcedure
\end{algorithmic}
\end{algorithm}

\noindent
We are now presenting details of the algorithm and its analysis.

\subsection{How to achieve polynomial delay and polynomial space}
In this section, we need a slight modification of the problem \textsc{ExtRD} in order to cope with pre-solutions. In this version, we add to an instance, usually specified by $G=\left( V, E\right)$ and $f:V\to \lbrace 0, 1, 2\rbrace$, a set $\overline{V_2}\subseteq V $ with $V_2\left(f\right)\cap \overline{V_2}=\emptyset$. The question is if there exists a minimal RDF $\widetilde{f}$ with $f \leq \widetilde{f}$ and $V_2\left(\widetilde{f}\right)\cap \overline{V_2}=\emptyset$. We call this problem a \emph{generalized} rdf extension problem, or  
\textsc{GenExtRD} for short.
In order to solve this problem, we  modify Algorithm \ref{alg} to cope with \textsc{GenExtRD} by adding an if-clause after Line~\ref{alg_if_no} that asks if $u\in\overline{V_2}$. If this is true, then the algorithm returns \emph{no}, because it is prohibited that $\tilde{f}(u)$ is set to~2, while this is necessary for minimal rdf, as there is a vertex $v$ in the neighborhood of $u$ such that $\tilde{f}(v)$ has been set to~1. We call this algorithm  \textsc{GenExtRD Solver}.

\begin{lemma}\label{lem:GenExtRD}
Let $G=\left(V,E\right)$ be a graph, $f : V \to \lbrace0,1,2\rbrace$ be a  function and $\overline{V_2} \subseteq V$ be a set with $V_2\left({f}\right)\cap \overline{V_2}=\emptyset$. \textsc{GenExtRD Solver} gives the correct answer when given the \textsc{GenExtRD} instance $(G,f,\overline{V_2})$.
\end{lemma}

\begin{pf}
In Algorithm~\ref{alg}, the only statement where we give a vertex the value $2$ is in the if-clause of Line~\ref{alg_if_no}. The modified version would first check if the vertex is in $\overline{V_2}$. If this is true, there will be no minimal RDF solving this problem. Namely, if we give the vertex the value~$2$, this would contradict $V_2 \left( \widetilde{f} \right) \cap \overline{V_2}=\emptyset$. If the value stays $1$, this would contradict Condition~\ref{con_1_2}. By \autoref{theorem:correctness_alg}, $\widetilde{f}$ will be a minimal rdf with $V_2\left(\widetilde{f}\right) \cap \overline{V_2}=\emptyset$ if the algorithm returns \emph{yes}.

Assume there exists a minimal RDF $\overline{f}$ with $V_2\left(\overline{f}\right) \cap \overline{V_2}=\emptyset$ but the algorithm returns  \emph{no}. First we assume that  \emph{no} is returned by the new if-clause. This implies that a vertex $u\in V_1\left( f \right)$ is in the neighborhood of a vertex $v\in V$ that has to have the value~2 in any minimal rdf that is bigger than $f$ (because \autoref{theorem:correctness_alg}). But this would lead to a similar contradiction as above.

Therefore, the answer \emph{no} has to be returned in Line \ref{alg_no}. That would contradict Condition~\ref{con_private} or Condition~\ref{con_min_dom}. Thus the algorithm would correctly return \emph{yes}.
\end{pf}

Let $f$ be a generalized rdf at any moment of the branching algorithm. The next goal is to show that \textsc{GenExtRD Solver} could tell us in polynomial time if there exists a minimal rdf that could be enumerated by the branching algorithm from this point on.   

\begin{proposition}
Let $G=\left(V,E\right)$ be a graph, $f : V \to \lbrace0,1,2,\overline{1},\overline{2}\rbrace$ be a partial function. Then,  \textsc{GenExtRD Solver} correctly answers if there exists some minimal rdf $g: V \to \lbrace0,1,2\rbrace$ that is consistent with $f$ when \textsc{GenExtRD  Solver} is given the instance $(G,\hat f,\overline{V_2}(f))$.
\end{proposition}

The following proof makes use of the grdf invariants presented above, which are only formally proved to hold in the next subsection, in \autoref{prop:invariants}.

\begin{pf} We have to show two assertions: (1) If \textsc{GenExtRD  Solver} answers \emph{yes} on the instance $(G,\hat f,\overline{V_2}(f))$, then there exists a minimal rdf $g: V \to \lbrace0,1,2\rbrace$ that is consistent with $f$. (2) If there exists a minimal rdf $g: V \to \lbrace0,1,2\rbrace$ that is consistent with $f$, then \textsc{GenExtRD  Solver} answers \emph{yes} on the instance $(G,\hat f,\overline{V_2}(f))$. 

\smallskip\noindent
\underline{ad (1)}: Assume  \textsc{GenExtRD  Solver} found a minimal rdf $g$ such that $\hat f\leq g$ and $V_2(g)\cap \overline{V_2}(f)=\emptyset$. Let $v\notin A(f)$. First assume that $g(v)=2$. Clearly, vertices in $V_2(f)=V_2(\hat f)$ do not get changed, as they cannot be made bigger. Hence, assume $v\notin V_2(f)$ exists with $g(v)=2$. As \textsc{GenExtRD Solver} will only explicitly set the value~2 for vertices originally set to~1 (by their $f$-assignment) that are in the neighborhood of vertices already set to value~2 and that do not belong to $\overline{V_2}(f)$, we have to reason about a possible $v\in V_1(f)=V_1(\hat f)$. 
By the third grdf invariant, the neighborhood of $v$ contains no vertex from $V_2(f)=V_2(\hat f)$, so that the case of some  $v\notin V_2(f)$  with $g(v)=2$ can be excluded.

Secondly, assume that $g(v)=1$. The case $v\in V_1(f)$ is not critical, and $v\in V_2(f)$ is not possible, as reasoned above. Notice that $g(v)=1$ was set in the last lines of the algorithm. In particular, $N_G(v)\cap V_2(g)=\emptyset$. As $V_2(f)\subseteq V_2(g)$, also $N_G(v)\cap V_2(f)=\emptyset$. By the first grdf invariant, 
$v\notin \overline{V_1}(f)\cup V_0(f)$. Hence, only $v\in \overline{V_2}(f)$ remains as a possibility.

Thirdly, assume that $g(v)=0$. As $f(v)\in\{1,2\}$ is clearly impossible, $v\in V_0(f)\cup \overline{V_1}(f)\cup \overline{V_2}(f)$ must follow.
Hence, $g$ is consistent with $f$.

\smallskip\noindent
\underline{ad (2)}: Assume that there exists a minimal rdf $g: V \to \lbrace0,1,2\rbrace$ that is consistent with $f$. We have to prove that $\hat f\leq g$ and that $V_2(g)\cap \overline{V_2}(f)=\emptyset$, because then \textsc{GenExtRD Solver} will correctly answer \emph{yes} by \autoref{lem:GenExtRD}.
For $g(v)=2$, then consistency implies $f(v)\neq \overline{2}$, and trivially $\hat f(v)\leq g(v)$. 
For $g(v)=1$, $v\in A(f)\cup V_1(f)\cup \overline{V_2}(f)$, and hence $\hat f(v)\in\{0,1\}$, so that $\hat f(v)\leq g(v)$. 
If $g(v)=0$, then $v\in A(f)\cup V_0(f)\cup \overline{V_1}(f)\cup \overline{V_2}(f)=V_0(\hat f)$, so that again  $\hat f(v)\leq g(v)$. 
\end{pf}

An important consequence of the previous proposition is stated next.
Notice that our algorithm behaves quite differently from what is known about algorithms the enumerate minimal ds.

\begin{proposition}\label{prop:poly-delay}
Procedure \textsc{Refined RD Enumeration}, on input $G=(V,E)$, outputs functions $f:V\to\{0,1,2\}$ with polynomial delay.
\end{proposition}

\begin{pf} Although the reduction rules are only stated in the next subsection, it is not hard to see by quickly browsing through them that they can be implemented to run in polynomial time. Moreover, \textsc{GenExtRD Solver} runs in polynomial time.
Hence, all work done in an inner node of the search tree needs polynomial time only.
By the properties of \textsc{GenExtRD Solver}, the search tree will never continue branching if no outputs are to be expected that are consistent with the current grdf (that is associated to that inner node). Hence, a run of the procedure \textsc{Refined RD Enumeration} dives straight through setting more and more values of a grdf, until it is everywhere defined with values from $\{0,1,2\}$, and then it returns from the recursion and dives down the next promising branch. Clearly, the length of any search tree branch is bounded by $|V|$, so that at most $2|V|$ many inner nodes are visited between any two outputs. This also holds at the very beginning (i.e., only polynomial time will elapse until the first function is output) and at the very end (i.e., only polynomial time will be spent after outputting the last function). 
This proves the claimed  polynomial delay.
\end{pf}

\begin{proposition}
Procedure \textsc{Refined RD Enumeration} correctly enumerates all minimal rdf that are consistent with the input grdf, assuming that at least one consistent rdf exists.
\end{proposition}

\begin{pf}
As there exists a consistent rdf, outputting the input function is correct if the input function is already an rdf, which is checked, as we test if the given grdf is everywhere defined at has only images in $\{0,1,2\}$. Also, before \textsc{Refined RD Enumeration} is called recursively, we  explicitly check if  at least one consistent rdf exists.

If the input grdf~$f$ is not everywhere defined or if $\overline{V_2}(f)\cup \overline{V_1}(f)\neq\emptyset$, then $A(f)\cup\overline{V_1}(f)\neq\emptyset$ by the fourth grdf invariant. Hence, whenever \textsc{Refined RD Enumeration} is called recursively, $A(f)\cup\overline{V_1}(f)\neq\emptyset$ holds, as these calls are immediately after applying all reduction rules exhaustively.

Hence, by induction and based on the previous propositions, \textsc{Refined RD Enumeration} correctly enumerates all minimal rdf that are consistent with the input grdf.
\end{pf}

\begin{corollary}\label{cor:correct-enumeration}
Procedure \textsc{Refined RD Enumeration} correctly enumerates all minimal rdf of a given graph $G=(V,E)$ when provided with the nowhere defined grdf $f_\bot$.
\end{corollary}

\begin{pf}
Due to the previous proposition, it is sufficient to notice that all minimal rdf are consistent with  $f_\bot$ and that the function that is constant~1 is a minimal rdf consistent with  $f_\bot$.
\end{pf}

\begin{proposition}\label{prop:poly-space}
Procedure \textsc{Refined RD Enumeration} never enumerates any minimal rdf consistent with the given grdf on the input graph $G=(V,E)$ twice.
\end{proposition}

\begin{pf}
Notice that the enumeration algorithm always branches by deciding for a vertex~$v$ from $A(f)\cup \overline{V_1}(f)$, where $f$ is the current grdf, if $f(v)$ is updated to~$2$ or not, which means that either $v\in A(f)$ is set to $\overline{2}$, or $v\in \overline{V_1}$ is set to~0. Then, reduction rules may apply, but they never change the decision if, in a certain branch of the search tree,  $f(v)=2$ is either true or false.
Moreover, they never set any vertex to~$2$. 
As any  minimal rdf that is ever output in a certain branch will be consistent with the grdf~$f$ associated to an inner node of the search tree, Procedure \textsc{Refined RD Enumeration} never enumerates any minimal rdf twice.
\end{pf}

An important consequence of the last claim is that there is no need to store all output 
functions in order to finally parse them to see into enumerating any of them only once.

\begin{corollary}\label{cor:poly-space}
Algorithm  \textsc{Refined RD Enumeration} lists all minimal rdf consistent with the given grdf on the input graph $G=(V,E)$ without repetitions and in polynomial space.
\end{corollary}

\subsection{Details on reductions and branchings}

For the presentation of the following rules, we assume that $G=(V,E)$ and a grdf $f$ is given. We also assume that the rules are executed exhaustively in the given order.

\vspace{5pt}
\noindent
{\bf Reduction Rule LPN (Last Potential Private Neighbor).} If $v\in V_2(f)$ satisfies $|N_G(v)\cap  (\overline{V_2}(f)\cup A(f))|=1$, then set $f(x) = 0$ for $\{x\}=N_G(v)\cap  (\overline{V_2}(f)\cup A(f))$.

\vspace{5pt}
\noindent
{\bf Reduction Rule $V_0$.} Let  $v \in V_0(f)$. Assume there exists a unique $u\in V_2(f)\cap N_G(v)$. 
Moreover, assume that for all $x\in N_G(u)\cap (V_0(f)\cup \overline{V_1}(f)\cup \overline{V_2}(f))$, $|N_G(x)\cap V_2(f)|\geq 2$ if $x\neq v$. Then, for any  $w \in N_G(v)\cap A(f)$, set $f(w) = \overline{2}$ and  for any  $w \in N_G(v)\cap \overline{V_1}(f)$, set $f(w) =0$.

\vspace{5pt}
\noindent
{\bf Reduction Rule $V_1$.} Let  $v \in V_1(f)$. For any  $w \in N_G(v)\cap A(f)$, set $f(w) = \overline{2}$.  For any  $w \in N_G(v)\cap \overline{V_1}(f)$, set $f(w) =0$.

\vspace{5pt}
\noindent
{\bf Reduction Rule $V_2$.} Let  $v \in V_2(f)$. For any  $w \in N_G(v)\cap A(f)$, set $f(w) = \overline{1}$.  For any  $w \in N_G(v)\cap \overline{V_2}(f)$, set $f(w) =0$.

\vspace{5pt}
\noindent
{\bf Reduction Rule NPD (No Potential Domination).} If $v\in \overline{V_2}(f)$ satisfies $N_G(v)\subseteq \overline{V_2}(f)\cup V_0(f)\cup V_1(f)$, then set $f(v) = 1$ (this also applies to isolated vertices in $\overline{V_2}(f)$).

\vspace{5pt}
\noindent
{\bf Reduction Rule NPN (No Private Neighbor).} If $v\in A(f)$ satisfies $N_G(v)\subseteq V_0 \cup\overline{V_1}(f)$, then set $f(v) = \overline{2}$ (this also applies to isolated vertices in $A(f)$).

\vspace{5pt}
\noindent
{\bf Reduction Rule Isolate.} If $A(f)=\emptyset$ and if $v\in \overline{V_1}(f)$ satisfies $N_G(v)\cap\overline{V_2}(f)= \emptyset$, then set $f(v) = 0$.

\vspace{5pt}
\noindent
{\bf Reduction Rule Edges.} If $u,v\in \overline{V_2}(f)\cup V_0(f)\cup V_1(f)$ and $e=uv\in E$, then remove the edge~$e$ from $G$.

\vspace{5pt}
\noindent
In the following, we first take care of the claimed grdf invariants. 

\begin{proposition}\label{prop:invariants}
After exhaustively executing the proposed reduction rules, as indicated in Algorithm~\ref{alg:refined-enum}, the claimed grdf invariants are maintained.
\end{proposition}

\begin{pf}We argue for the correctness of the grdf invariants  by induction one by one. Notice that (trivially) all invariants hold if we start the algorithm with the nowhere defined grdf.
\begin{enumerate}
    \item $\forall x\in \overline{V_1}(f)\cup V_0(f)\,\exists y\in N_G(x):y\in V_2(f)$.
    
    We need to show  that $N_G(x) \cap V_2(f) \neq \emptyset$ holds for each $x\in V_0(f) \cup \overline{V_1}(f)$. For the inductive step, we only have to look at the reduction rules, since the branching rules only change the value to $0$ if the vertex was already in $\overline{V_1}(f)$. For each reduction rule where we set a value of a vertex to~$0$ or to~$\overline{1}$, there exists a vertex in the neighborhood with value~$2$, which is seen as follows.

\begin{tabular}{rl}
    LPN:& We explicitly consider $x\in N_G(V_2(f))$ only to be set by $f(x)=0$.\\
    $V_0$\,\&\,$V_1$:& We only set $w$ to $0$ if it has been in $\overline{V_1}$. By induction  hypothesis,  \\
    &$w$ has a neighbor in $V_2(f)$.\\
    $V_2$:& We explicitly consider $w\in N_G(V_2(f))$ only to be set to~$0$ or to~$\overline{1}$.\\
    Isolate:& Only vertices from $\overline{V_1}(f)$ are set to~0; apply induction hypothesis. 
\end{tabular}
    
    \item $\forall x\in V_2(f):N_G(x)\subseteq \overline{V_1}(f)\cup V_0(f) \cup V_2(f)$.
    
    This property can only be invalidated if new vertices get the value~2 or if vertices from  $\overline{V_1}(f)\cup V_0(f) \cup V_2(f)$ are changed to a value other than this or if edges are deleted (as vertices are never deleted). The only way in which a vertex gets~$v$ the value~2 is by branching. Immediately afterwards, the reduction rules are executed: LPN and $V_2$ will install the invariant for the neighborhood of~$v$. No reduction rule ever changes the value of a vertex from $V_0(f) \cup V_2(f)$, while vertices from $\overline{V_1}(f)$ might be set to~$0$ or~$2 $. The Reduction Rule Edges deletes no edges incident to vertices from $V_2(f)$. 
    \item $\forall x\in V_1(f):N_G(x)\subseteq \overline{V_2}(f)\cup V_0(f)\cup V_1(f)$.
    
    The invariant is equivalent to the following  three conditions: (a) $N(V_1(f))\cap A(f)=\emptyset$, (b) $N(V_1(f))\cap \overline{V_1}(f)=\emptyset$ and (c) $N(V_1(f))\cap V_2(f) =\emptyset$. Conditions (a) and (b) are taken care of by Reduction Rule $V_1$.
    Condition (c) immediately follows by the already proven second invariant.
    \item If $\overline{V_2}(f)\neq\emptyset$, then $A(f)\cup \overline{V_1}(f)\neq \emptyset$.
    
    Consider some $x\in \overline{V_2}(f)$. By the second invariant, $N_G(x)\cap V_2(f)=\emptyset$. By the Reduction Rule Edges, $N_G(x)\cap \left(\overline{V_2}(f)\cup V_0(f)\cup V_1(f)\right)=\emptyset$. As the Reduction Rule NPD did not apply, the only possible neighbors of $x$ are in $A(f)\cup \overline{V_1}(f)$.\qed
\end{enumerate}
\renewcommand{\qed}{}
\end{pf}

We have now to show the \emph{soundness} of the proposed reduction rules. In the context of enumerating minimal rdf, this means the following: if $f,f'$ are grdf of $G=(V,E)$ before or after applying any of the reduction rules, then $g$ is a minimal rdf that is consistent with $f$ if and only if it is consistent with $f'$.

\begin{proposition}\label{prop:rdf-reductionrules-soundness}
All proposed reduction rules are sound.
\end{proposition}

\begin{pf}
For the soundness of the reduction rules, we also need the invariants proven to be correct in \autoref{prop:invariants}.
We now prove the soundness of each reduction rule, one at a time.

If possible, we apply Reduction Rule LPN first. Consider $v\in V_2(f)$ with $\{x\}=N_G(v)\cap  (\overline{V_2}(f)\cup A(f))$. Before the branching step, due to the second invariant, neighbors of $V_2(f)$-vertices are either in $V_0(f)$, $V_2(f)$ or in $\overline{V_1}(f)$. As no reduction rule adds a vertex to $V_2(f)$, $v$ must have been put into $V_2(f)$ by the last branching step.
By the first invariant, we know that all $y\in N_G(v)\cap\left( \overline{V_1}(f)\cup V_0(f)\right)$ are dominated by vertices different from~$v$. As  $v\in V_2(f)$, it still needs a private neighbor to dominate. As $N_G(v)\cap  (\overline{V_1}(f)\cup A(f))$ contains one element~$x$ only, setting $f(x)=0$ is enforced for any minimal rdf (see \autoref{t_private_neighborhood}).

Next, we prove Reduction Rule $V_0$. We consider $v\in V_0(f)$ and $u\in V$ with $\{u\}=V_2(f)\cap N_G(v)$. We can use the rule, since $u$ needs a private neighbor which can only be $v$, by the assumption that every other neighbor of~$u$ is dominated at least twice. To maintain the property that $v$ is a private neighbor of~$u$, each $A(f)$-neighbor of~$v$ is set to~$\overline{2}$ and each $\overline{V_1}(f)$-neighbor of $v$ is set to~$0$. This annotates the fact that any minimal rdf~$g$ compatible with~$f$ will satisfy $(g(x)=2\implies x=u)$ for each $x\in N_G(v)$.

The soundness of Reduction Rule $V_1$ and Reduction Rule $V_2$ mainly follows from \autoref{t_1_2_neigborhood}. 

Coming to the Reduction Rule NPD, notice that setting $f(v)=0$ would necessitate $f(u)=2$ for some neighbor $u$ of $v$, which is impossible.

For the Reduction Rule NPN, we use the fact that $N_G(v) \cap V_2(f) \neq \emptyset$ holds for each $v\in V_0(f) \cup \overline{V_1}(f)$, which is the first invariant.\footnote{More precisely, we also have to check that the possibly newly introduced vertices in $V_0(f)$ or $\overline{V_1}(f)$ by the branching or by the reduction rules up to this point do maintain the invariant, but this is nothing else then re-checking the induction step of the correctness proof of this invariant, see the proof of \autoref{prop:invariants}.}  This implies that  there is no element left for $v$ to dominate (therefore it has no private neighbor except itself). Thus, if $v$ has the value $2$, then it would contradict with \autoref{t_private_neighborhood}.

For the soundness of Reduction Rule Isolate, we note that, since $A(f)$ is empty, $v\in \overline{V_1}(f)$ can only have neighbors in $V_1(f)\cup V_2(f)\cup \overline{V_1}(f)$, as $\overline{V_2}(f)$-neighbors are prohibited. As  Reduction Rule $V_1$ was (if possible) executed before, $N_G(v)\cap V_1(f)= \emptyset$ holds. Therefore, $v\in \overline{V_1}(f)$ would not have a private neighbor if $f(v)=2$, cf. the first invariant.\footnote{Again, one has to partially follow the induction step of the proof of \autoref{prop:invariants}.}

Finally, the soundness of Reduction Rule Edges follows trivially from the fact that an element of $\overline{V_2}(f) \cup V_0(f) \cup V_1(f)$ cannot dominate any vertex in $\overline{V_2}(f) \cup V_0(f) \cup V_1(f)$. Hence, a minimal rdf $g$ is consistent with $f$ if and only it is consistent with $f'$, obtained by applying the Reduction Rule Edges to~$f$.
\end{pf}

In order to fully understand Algorithm~\ref{alg:refined-enum}, we need to describe priorities for branching.
We describe these priorities in the following in decreasing order for a vertex $v\in A(f)\cup \overline{V_1}(f)$.
\begin{enumerate}
    \item $v\in A(f)$ and $|N_{G}(v)\cap (A(f)\cup \overline{V_2}(f))|\geq 2$;
    \item any $v\in A(f)$;
    \item any $v\in \overline{V_1}(f)$, preferably if $|N_{G}(v)\cap \overline{V_2}(f)|\neq 2$.
\end{enumerate}

These priorities also split the run of our algorithm into phases, as whenever the algorithm was once forced to pick a vertex according to some lower priority, there will be never again the chance to pick a vertex of higher priority thereafter.
It is useful to collect some \textbf{phase properties} that instances must satisfy after leaving Phase~$i$, determined by applying the $i^\text{th}$ branching priority.

\begin{itemize}
\item Before entering any phase, there are no edges between vertices $u,v$ if $u,v\in V_0(f)\cup V_1(f)\cup \overline{V_2}(f)$ or if $u\in V_2(f)$ and $v\in \overline{V_2}(f)\cup A(f)$ or if $u\in V_1(f)$ and  $v\in \overline{V_1}(f)\cup A(f)$, as we can assume that the reduction rules have been exhaustively applied.
    \item After leaving the first phase, any active vertex with an active neighbor is either pendant or has only further neighbors from $\overline{V_1}(f)\cup V_0(f)$.
    
    \item After leaving the second phase, $A(f)=\emptyset$ and $N_G(\overline{V_2}(f))\subseteq \overline{V_1}(f)$. Moreover, any vertex $x\in \overline{V_2}(f)$ has neighbors in $\overline{V_1}(f)$.
    \item After leaving the third phase, $A(f)=\overline{V_2}(f)=\overline{V_1}(f)=\emptyset$, so that $f$ is a Roman dominating function. 
\end{itemize}

\begin{proposition}
The phase properties hold.
\end{proposition}

\begin{proof} We are considering the items on the list separately.
\begin{itemize}\item Reduction Rule Edges shows the first claim. Reduction Rules $V_2$ and $V_1$ show the other two claims.
\item 
By the branching condition, we know that after leaving the first phase, $|N_{G}(v)\cap (A(f)\cup \overline{V_2}(f))|< 2$ for any active vertex~$v$. Since $v$ has a neighbor in $A(f)$ (say $u$) this implies that there cannot be any other neighbor in $A(f)\cup \overline{V_2}(f)$. Moreover, by the Reduction Rule $V_1$, $N_G(v)\cap V_1(f)=\emptyset$, and by the Reduction Rule $V_2$, $N_G(v)\cap V_2(f)=\emptyset$. Hence, $N_G(v) \setminus \lbrace u\rbrace \subseteq \overline{V_1}(f)\cup V_0(f)$.
\item The second phase branches on each $v\in A(f)$. Therefore, it ends if $A(f)=\emptyset$.  Let $v\in \overline{V_2}(f)$. By Reduction Rule Edge, we get $N_G(v) \cap \left( \overline{V_2}(f) \cup V_0 \cup V_1 \right)=\emptyset$. Reduction Rule $V_2$ implies that $v$ does not have a neighbor in $V_2(f)$. Therefore we get $N_G(v) \subseteq \overline{V_1}(f)$. If $N_G(v)$ is empty, Reduction Rule NPD will be triggered. Therefore, $N_G(v)$ has at least one element.
\item The third phase runs on the vertices in $\overline{V_1}(f)$. Thus, $\overline{V_1}(f)=\emptyset$ holds at the end of this phase. Since we never put a vertex into $A(f)$ again, $A(f)$ is empty. To get $\overline{V_2}(f) = \emptyset$, we can use the same argumentation as in the property before, since a vertex goes only from $A(f)$ to $\overline{V_2}(f)$.\qed 
\end{itemize}
\end{proof}

\subsection{A Measure \& Conquer Approach}

We now present the branching analysis, classified by the described branching priorities.
We summarize a list of all resulting branching vectors in \autoref{tab:branching-vectors}.

\begin{table}[tb]
    \centering
    \begin{tabular}{|l|l|}\hline
         Phase \#&  Branching vector\\\hline 
         1.1 & $(1-\omega_2,3-2\omega_1)$\\
         1.2 & $(1-\omega_2,1+2\omega_2)$\\
         1.3 & $(1-\omega_2,2+\omega_2-\omega_1)$\\\hline
         2.1 \& 2.2.b & $(1-\omega_2,2)$\\
         2.2.a & $(1-\omega_2,1+\omega_2+\omega_1)$\\
         2.2.c & $(1+\omega_2,1)$\\\hline
         3.1 & $(\omega_1,\omega_1+3\omega_2)$\\
         3.2.a & $(\omega_1,2\omega_1+\omega_2)$\\
         3.2.b \& 3.3.a & $(\omega_1+\omega_2,\omega_1+\omega_2)$\\3.3.b & $(2\omega_1+2\omega_2,2\omega_1+2\omega_2,2\omega_1+2\omega_2,2\omega_1+2\omega_2)$\\\hline
    \end{tabular}
    \caption{The branching vectors of different branching scenarios of the enumeration algorithm for listing all minimal Roman domination functions of a given graph}
    \label{tab:branching-vectors}
\end{table}

\subsubsection{Branching in Phase~1.}
We are always branching on an active vertex $v$. In the first branch, we set $f(v) = 2$.
In the second branch, we set $f(v) = \overline{2}$.
In the first branch, in addition the Reduction Rule $V_2$ triggers at least twice. In order to determine a lower bound on the branching vector, we describe three worst-case scenarios; all other reductions of the measure can be only better.
\begin{enumerate}
    \item $|N_G(v)\cap A(f)|=2$, i.e., $v$ has two active neighbors $x$ and $y$.  The corresponding recurrence is: $T(\mu) = T(\mu-(1-\omega_2))+T(\mu-(1+2(1-\omega_1)))$, as either $v$ moves from $A(f)$ to $V_2(f)$ and $x,y$ move from $A(f)$ to $\overline{V_1}(f)$, or $v$ itself moves from $A(f)$ to $\overline{V_2}(f)$. The branching vector is hence: $(1-\omega_2,3-2\omega_1)$, as noted in the first row of \autoref{tab:branching-vectors}.
    \item $|N_G(v)\cap \overline{V_2}(f)|=2$. The corresponding recurrence is:  $T(\mu) = T(\mu-(1-\omega_2))+T(\mu-(1+2\omega_2))$, see the second row of \autoref{tab:branching-vectors}.
    \item $|N_G(v)\cap A(f)|=1$ and $|N_G(v)\cap \overline{V_2}(f)|=1$, leading to  $T(\mu) = T(\mu-(1-\omega_2))+T(\mu-(1+(1-\omega_1)+\omega_2))$, see \autoref{tab:branching-vectors}, third row.
\end{enumerate}

\subsubsection{Branching in Phase~2.}
We are again branching on an active vertex $v$. By Reduction Rule NPN, we can assume that $N_G(v)\neq\emptyset$. In the first branch, we set $f(v) = 2$.
In the second branch, we set $f(v) = \overline{2}$.

\begin{enumerate}
    \item
If $N_G(v)\cap A(f)=\{x\}$, then $N_G(v)\cap \overline{V_2}(f)=\emptyset$ in this phase. Therefore, in the first branch, $f(x)=0$ is enforced by Reduction Rule LPN. Notice that this might further trigger Reduction Rule $V_0$ if $N_G(x)\cap (A(f)\cup \overline{V_1}(f))$ contains vertices other than~$v$. The corresponding worst-case recurrence is: $T(\mu) = T(\mu-(1-\omega_2))+T(\mu-(1+1))$, see \autoref{tab:branching-vectors}, fourth row.
    \item
If $N_G(v)\cap  \overline{V_2}(f)=\{x\}$, then $N_G(v)\cap A(f)=\emptyset$ in this phase. Therefore, in the first branch, $f(x)=0$ is enforced by Reduction Rule LPN. We consider several sub-cases now.
\begin{enumerate}
    \item $N_G(x)\cap \overline{V_1}(f)\neq\emptyset$. Reduction Rule $V_0$ will put all these vertices into $V_0(f)$. The corresponding worst-case recurrence is: $T(\mu) = T(\mu-(1-\omega_2))+T(\mu-(1+\omega_2+\omega_1))$, see \autoref{tab:branching-vectors}, fifth row.
    \item $|N_G(x)\cap A(f)|\geq 2$. Reduction Rule $V_0$ will put all these vertices into $\overline{V_2}(f)$  (except for $v$). The corresponding worst-case recurrence is: $T(\mu) = T(\mu-(1-\omega_2))+T(\mu-(1+\omega_2+(1-\omega_2)))$, see \autoref{tab:branching-vectors}, fourth row.
    \item Recall that by Reduction Rule Edges, $N_G(x)\cap (\overline{V_2}(f)\cup V_0(f)\cup V_1(f))=\emptyset$, so that (if the first two cases do not apply) now we have $N_G(x)\setminus\{v\}\subseteq V_2(f)$. By the properties listed above, also $N_G(x)\cap V_2(f)=\emptyset$ is clear, so that now $|N_G(x)|=1$, i.e., $x$ is a pendant vertex. In this situation, we do not gain anymore from the first branch, but when $f(v)=\overline{2}$ is set, Reduction Rule NPD triggers and sets $f(x)=1$. The corresponding worst-case recurrence is: $T(\mu) = T(\mu-(1+\omega_2))+T(\mu-(1-\omega_2+\omega_2))$, see \autoref{tab:branching-vectors}, sixth row.
\end{enumerate}
\end{enumerate}

\subsubsection{Branching in Phase~3.}
As $A(f)=\emptyset$, we are now branching on a vertex $v\in \overline{V_1}(f)$. Due to Reduction Rule Isolate, we know that $N_G(v)\cap\overline{V_2}(f)\neq \emptyset$.
In the first branch, we consider setting $f(v)=2$, while in the second branch, we set $f(v)=0$.
Again, we discuss several scenarios in the following.

\begin{enumerate}
    \item Assume that $|N_G(v)\cap \overline{V_2}(f)|\geq 3$. If we set $f(v)=2$, then Reduction Rule $V_2$ triggers at least thrice. The corresponding worst-case recurrence is: $T(\mu) = T(\mu-\omega_1)+T(\mu-(\omega_1+3\omega_2))$, with a branching vector of $(\omega_1,\omega_1+3\omega_2)$, see \autoref{tab:branching-vectors}, seventh row.  
    \item Assume that $|N_G(v)\cap \overline{V_2}(f)|=1$, i.e., there is some (unique) $u\in \overline{V_2}(f)$ such that $N_G(v)\cap \overline{V_2}(f)=\{u\}$. We consider two sub-cases:
    \begin{enumerate}
        \item If  $|N_G(u)\cap \overline{V_1}(f)|\geq 2$, then if we set $f(v)=2$, then first Reduction Rule LPN triggers $f(u)=0$, which in turn sets $f(w)=0$ for all $w\in N_G(u)\cap \overline{V_1}(f)$, $w\neq u$, by Reduction Rule $V_0$. The corresponding worst-case recurrence is: $T(\mu) = T(\mu-\omega_1)+T(\mu-(\omega_2+2\omega_1))$, see \autoref{tab:branching-vectors}, eighth row.  
        \item If  $|N_G(u)\cap \overline{V_1}(f)|=1$, then $u$ is a pendant vertex. Hence, in the first branch, we have (as above) $f(v)=2$ and $f(u)=0$, while in the second branch, we have $f(v)=0$ and $f(u)=1$ by Reduction Rule NPD. This decreases the measure by $\omega_1+\omega_2$ in both branches, see \autoref{tab:branching-vectors}, nineth row. This scenario happens in particular if the graph $G'$ induced by $\overline{V_1}(f)\cup \overline{V_2}(f)$ contains a connected component which is a $P_2$. Therefore, we refer to this (also) as a \emph{$P_2$-branching} below.
    \end{enumerate}
    \item Assume that $|N_G(v)\cap \overline{V_2}(f)|=2$, i.e., there are some $u_1,u_2\in \overline{V_2}(f)$ such that $N_G(v)\cap \overline{V_2}(f)=\{u_1,u_2\}$. Notice that in the first branch, when $f(v)=2$, Reduction Rule $V_2$ triggers twice, already reducing the measure by  $\omega_1+2\omega_2$.
    We consider further sub-cases:
    \begin{enumerate}
        \item If  $|N_G(u_1)\cap \overline{V_1}(f)|=1$, then $u_1$ is a pendant vertex. As in the previous sub-case, this helps us reduce the measure in the second branch   by $\omega_1+\omega_2$ due to Reduction Rule NPD, which obviously puts us in a better branching than \autoref{tab:branching-vectors}, nineth row. Similarly, we can discuss the case  $|N_G(u_2)\cap \overline{V_1}(f)|=1$.
        
        \item If  $|N_G(u_1)\cap \overline{V_1}(f)|=2$, then we know now that the graph $G'$ induced by $\overline{V_1}(f)\cup \overline{V_2}(f)$ is  bipartite after removing edges between vertices from $\overline{V_1}(f)$,  and vertices from $\overline{V_1}(f)$ all have degree two and vertices from $\overline{V_2}(f)$ all have degree at least two. The worst case for the following branching is hence given by a $K_{2,2}$ as a connected component in $G'$: Testing now all possibilities of setting the $\overline{V_2}(f)$-vertices to $0$ or to $1$ will determine all values of the $\overline{V_1}(f)$-vertices by reduction rules. Hence, we have in particular for the  $K_{2,2}$ a scenario with four branches, and in each branch, the measure is reduced by $2\omega_2+2\omega_1$ \emph{($K_{2,2}$-branching)}.
      
    \end{enumerate}
\end{enumerate}

\begin{proposition}\label{prop:run-time}On input graphs of order~$n$, 
Algorithm \textsc{Refined RD Enumeration} runs in time $\Oh^*(\RomanUpperbound^n)$.
\end{proposition}

\begin{pf}
We follow the run-time analysis that led us to the branching vectors listed in \autoref{tab:branching-vectors}. The claim follows by choosing as weights $\omega_1= \frac{2}{3}$ and $\omega_2=0.38488$.
\end{pf}

The two worst-case branchings (with the chosen weights $\omega_1= \frac{2}{3}$ and $\omega_2=0.38488$) are 1.1, 3.2.b and 3.3.
If we want to further improve on our figures, we would have to work on a deeper analysis in these cases.
For the $P_2$-branching, it might be an idea to combine it with the branchings where it could ever originate from. Notice that adjacent $\overline{V_1}$-$\overline{V_2}$-vertices can be only produced in the first branching phase. But we would then have to improve also on Phase 3.3, the worst case being a $K_{2,2}$-branching in Case 3.3 (b).

\smallskip\noindent
Let us finally summarize the corner-stones of our reasoning.
\begin{proof}[\autoref{thm:minimal-rdf-enumeration}]
Several important properties have been claimed and proved about Algorithm~\ref{alg:refined-enum} that show the claim of our second main theorem.
\begin{itemize}
    \item The algorithm correctly enumerates all minimal rdf; see \autoref{cor:correct-enumeration}.
    \item The algorithm needs polynomial space only; see \autoref{cor:poly-space}.
    \item The algorithm achieves polynomial delay; see \autoref{prop:poly-delay}.
    \item The algorithm runs in time $\Oh^*(\RomanUpperbound^n)$ on input graphs of order~$n$; see \autoref{prop:run-time}.\qed 
\end{itemize}
\end{proof}

\section{An Alternative Notion of Minimal RDF}
\label{sec:alternative-notion}

So far, we focused on an ordering of the functions $V\to\{0,1,2\}$ that was derived from the linear ordering $0<1<2$. Due to the different functionalities, it might be not that clear if 2 should be bigger than 1.
If we rather choose as a basic partial ordering $0<1,2$, with $1,2$ being incomparable, this yields another ordering for the functions $V\to\{0,1,2\}$, again lifted pointwise. Being reminiscent of partial orderings, let us call the resulting notion of minimality PO-minimal rdf.
Recall that the notion of minimality for Roman dominating functions that we considered so far and that we also
view as the most natural interpretation of this notion has been refuted in the literature, because it leads to a trivial notion of \textsc{Upper Roman Domination}, because the minimal rdf $f:V\to\{0,1,2\}$ with biggest
sum $\sum_{v\in V}f(v)$ is achieved by the constant function $f=1$. This is no longer true for the (new) problem \textsc{Upper PO-Roman Domination}. 

Also, this can be seen as a natural pointwise lifting of the inclusion ordering, keeping in mind that $f\leq_{PO}g$ iff $V_1(f)\subseteq V_1(g)$ and $V_2(f)\subseteq V_2(g)$.

More interesting for the storyline of this paper are the following results:

\begin{theorem}\label{t_porperty_min_rdf}
Let $G=\left(V,E\right)$ be a graph, $f: \: V \to \lbrace 0,1,2\rbrace$  
and abbreviate
$G'\coloneqq G\left[ V_0\left(f\right)\cup V_2\left(f\right)\right]$. Then, $f$ is a PO-minimal rdf if and only if the following conditions hold:
\begin{enumerate}
\item$N_G\left[V_2\left(f\right)\right]\cap V_1\left(f\right)=\emptyset$,\label{con_1_2_PO}
\item $V_2\left(f\right)$ is a minimal dominating set of $G'$.\label{con_min_dom_PO}
\end{enumerate}
\end{theorem}

\begin{pf}
First we look into the ``only if''-part. The first condition follows analogously from \autoref{t_1_2_neigborhood}. For the other condition, we assume that there exists a graph $G=\left(V,E\right)$ and a PO-minimal-rdf $f: V \to \lbrace 0,1,2\rbrace$ such that $V_2(f)$ is not a minimal dominating set in $G'$. Since $f$ is a rdf, $V_2(f)$ is a dominating set in~$G'$. Thus, $V_2(f)$ is not irredundant in $G'$. Hence, there exists a $v\in V_2(f)$ such that $N[v] \subseteq N_G[V_2(f)\setminus \lbrace v\rbrace]$. Define 
$$ \widetilde{f}:V\to \lbrace 0,1,2\rbrace,\: w\mapsto\begin{cases}
f\left(w\right), &w\neq v\\
0,& w=v
\end{cases}. $$
Clearly, vertices $w\in \left( V_0(f) \cup V_1(f)\right)\setminus N_G[v]$ are dominated by $\widetilde{f}$. But $N_G[v]$ is also dominated, since $N_G[v] \subseteq N_G[V_2(f)\setminus \lbrace v\rbrace]$ holds. This would contradict the PO-minimality of $f$. 

Let $f$ be a function that fulfills the two conditions. Since $V_2\left(f\right)$ is a dominating set in $G'$, for each $u\in V_0\left( f\right)$, there exists a $v\in V_2\left(f\right)\cap N_{G}\left[u\right]$. Therefore, $f$ is a rdf. 
Let $\widetilde{f}:V \to \lbrace 0,1,2 \rbrace$ be a PO-minimal rdf such that $\widetilde{f}$ is smaller than $f$ with respect to the partial ordering. Therefore, $\widetilde{f}$ (also) satisfies the two conditions. 
Assume that there exists a $v\in V$ with $\widetilde{f}\left( v\right) < f\left( v \right)$. Hence,  $V_2\left(\widetilde{f}\right)\subseteq V_2\left(f\right)\setminus \lbrace v\rbrace$.  

\noindent
\textbf{Case 1:} $\widetilde{f}\left( v\right)=0 ~and~ f\left( v \right) =1$. Therefore, there exists a vertex $u\in N_G\left(v\right)$ with $f\left(u\right)\geq \widetilde{f}\left(u\right)=2$. This contradicts Condition~\ref{con_1_2}.

\noindent
\textbf{Case 2:} $\widetilde{f}\left( v\right)=0 ~and~ f\left( v \right) =2$. Thus, for each $u\in V_0(f) \cap N_G[v]\subseteq V_0(\widetilde{f})\cap N_G[v]$ there exists a $w\in V_2(\widetilde{f})\cap N_G(u)\subseteq V_2(f)\cap N_G(u)$. This implies, that $V_2(f)$ is not irredundant in $G'$, which contradicts the second condition.

Therefore, $\widetilde{f} = f$ holds and $f$ is PO-minimal.
\end{pf}

Based on this characterization of PO-minimality, we can again derive a positive algorithmic results for the corresponding extension problem.

\begin{theorem}\label{theorem:P_ExtPoRDF}
The extension problem \textsc{ExtPO-RDF} can be solved in polynomial time.
\end{theorem}

\begin{pf}
For this problem, we have to modify Algorithm \ref{alg} again. This time, we have to modify Line \ref{alg_private_test} to \textbf{if} $N_G[v]\subseteq N[M_2\setminus \lbrace v\rbrace]$ \textbf{do}. The rest of the proof is analogous to the proof of  \autoref{theorem:correctness_alg}.
\end{pf}

Furthermore, we can show that for PO-minimal rdf, the simple enumeration algorithm is already provably optimal.

\begin{theorem}
There is a polynomial-space algorithm that enumerates all PO-minimal rdf of a given graph of order $n$ in time $\Oh^*(2^n)$ with polynomial delay. Moreover, there is a family of graphs $G_n$, with $G_n$ being of order $n$, such that $G_n$ has $2^n$ many PO-minimal rdf.
\end{theorem}

\begin{pf}
The algorithm itself works similar to Algorithm~\ref{alg:enum}, but we have to integrate the extension tests as in Algorithm~\ref{alg:refined-enum}. Therefore, we need to combine our two modifications for Algorithm~\ref{alg}. This new version would solve the \textsc{GenExtPO-RDF}, where a graph $G=(V,E)$, a function $f:V\to\lbrace0,1,2\rbrace$ and a set $\overline{V_2}$ are given and we need to find a PO-minimal rdf $\widetilde{f}$ with $f\leq \widetilde{f}$ and $\overline{V_2}\cap V_2\left( f\right)=\emptyset$ (to prove this, combine the proofs of \autoref{lem:GenExtRD} and \autoref{theorem:P_ExtPoRDF}). 
To see optimality of the enumeration algorithm, notice that the null graph (edge-less graph) of order~$n$ has any mapping $f:V\to\{1,2\}$ as a PO-minimal rdf. 
\end{pf}

It follows that the (relatively simple) enumeration algorithm is optimal for PO-minimal rdf.
If one dislikes the fact that our graph family is disconnected, consider the star $K_{1,n}$ that has $2^{n}+1$ many different PO-minimal rdf: If $V(K_{1,n})=\{0,1,\dots,n\}$, with $0$ being the center and $i\in\{1,\dots,n\}$ being the `ray vertices' of this star, then either put $f(0)=2$ and $f(i)=0$ for $i\in\{1,\dots,n\}$, or $f(j)=1$ for $j\in \{0,1,\dots,n\}$, or $f(0)=0$ and $f(i)\in\{1,2\}$ is arbitrary for  $i\in\{1,\dots,n\}$ (except for $f(j)=1$ for $j\in \{1,\dots,n\}$).
This example proves that there cannot be any general enumeration algorithm running in time  $\Oh((2-\varepsilon)^n)$ for any $\varepsilon>0$, even for connected graphs of order~$n$.

\section{Conclusions}

While the combinatorial concept of Roman domination leads to a number of complexity results that are completely analogous to what is known about the combinatorial concept of  domination, the two concepts lead to distinctively different results when it comes to enumeration and extension problems. These are the main messages and results of the present paper.

We are currently working on improved enumeration and also on counting of minimal rdf in special graph classes.
Our first results are very promising; for instance, there are good chances to completely close the gap between lower and upper bounds for enumerating minimal rdf for some graph classes.

Another line of research is looking into problems that are similar to Roman domination, in order to better understand the specialties of Roman domination in contrast to the classical domination problem. What makes Roman domination behave different from classical domination when it comes to finding extensions or to enumeration?

Finally, let us mention that our main branching algorithm also gives an input-sensitive enumeration algorithm for minimal Roman dominating functions in the sense of Chellali \emph{et al.}~\cite{CheHHHM2016}. However, we do not know of a polynomial-delay enumeration algorithm in that case. This is another interesting line of research.
Here, the best lower bound we could find was a repetition of a $C_4$, leading to $\sqrt[4]{8}\geq 1.68179$ as the basis.

\end{document}